\documentclass[12pt]{huber_article}
\usepackage[utf8]{inputenc}

\usepackage{amsthm,amsmath,amsfonts}
\usepackage{booktabs}
\usepackage{comment}
\usepackage[margin=1.25in,footskip=0.25in]{geometry}

\usepackage{tikz}

\newcommand{\mean}{\mathbb{E}}
\newcommand{\prob}{\mathbb{P}}

\newcommand{\prar}{{\sffamily{PRAR}}}
\newcommand{\ar}{{\sffamily{AR}}}
\newcommand{\mcmc}{{\sffamily{MCMC}}}

\newcommand{\berndist}{{\textsf{Bern}}}

\newcommand{\gk}{{\arabic{line})}\stepcounter{line}}
\newcounter{line}

\newtheorem{lemma}{Lemma}

\begin{document}

\title{Partially Recursive Acceptance Rejection}
\author{\noindent Mark Huber\\ Claremont McKenna College \\ {\tt mhuber@cmc.edu}}

\maketitle

\begin{abstract}
  Generating random variates from high-dimensional distributions is often
  done approximately using Markov chain Monte Carlo.  In certain cases,
  perfect simulation algorithms exist that allow one to draw exactly from
  the stationary distribution, but most require $O(n \ln(n))$ time, where
  $n$ measures the size of the input.  In this work a new protocol for 
  creating perfect simulation algorithms that runs in $O(n)$ time for a wider
  range of parameters on several models (such as Strauss, Ising, and random cluster)
  than was known previously.  This
  work represents an extension of the popping algorithms due to Wilson.
\end{abstract}

\thispagestyle{empty}

\setcounter{page}{1}

\section{Introduction}
\label{SEc:introduction}

Partially recursive acceptance rejection (\prar{}) is a new protocol for creating algorithms for exactly generating random variates from high dimensional distributions.  The method is simple to implement efficiently, and results in algorithms that can be proved to have an expected running time that is linear in problem size over a wider range of parameters than known previously for several problems of interest.

Consider a distribution defined either by an unnormalized weight function $w(x)$ for discrete spaces, or an unnormalized density function $f(x)$ for continuous spaces.  In the past, techniques for perfect simulation from these distributions have been very different depending on whether the spaces were continuous or discrete.  \prar{} operates the same way in both situations.  

The problems where \prar{} is useful often use approximate sampling via Markov chain Monte Carlo (\mcmc{}).  However, these are not true algorithms unless the mixing time of the Markov chain being used can be somehow bounded.  One such technique for bounding the mixing time of a Markov chain is Dobrushin uniqueness~\cite{dobrushin1968} (see also the work on path coupling of Bubley and Dyer~\cite{bubleyd1997c}).  \prar{} requires a condition very similar to Dobrushin uniqueness, the difference being that even if Dobrushin uniqueness cannot be shown mathematically, it can be verified algorithmically.  

One can prove (up to an arbitrarily small chance of error) that a Markov chain is slowly mixing 
through computer experiment, however, there is no generally effective way to show that a Markov chain is rapidly mixing.  However, with perfect simulation protocols like \prar{}, it is possible through computer experiment to verify that the resulting output comes exactly from the target distribution, thereby making them more useful in practice for certain problems than the Markov chain approach.

Our concern here is with fast algorithms that use an expected number of random choices that is linear in the size of the problem.  There are currently two protocols for generating samples in expected linear time from these types of problems.
The first method is the {\em Clan of Ancestors} method of~\cite{ferrarifg2002}.  This method is an extension
of the Coupling from the Past~\cite{proppw1996} protocol or Propp and Wilson.  It looks backwards in time at
the events that could affect the state in the present.  The set of backwards events is called the Clan of 
Ancestors, which gives the name to the method.  This approach has rarely been
implemented (see~\cite{huber2012d} for one such implementation) because unless special care is taken it can
be very computationally expensive to keep track of how clans arising from different dimensions interact.  
One of the advantages of \prar{} is that such interactions cannot occur with this approach, making the 
implementation much easier.

The second linear time method is the {\em Randomness Recycler} method of~\cite{huber2000a,huber2000b}.  Like \prar{} (and
unlike Clan of Ancestors) it is straightforward to implement, however, the range of parameters for which 
the algorithm is provably linear time is in all known
cases more restricted than with \prar{}.

The \prar{} protocol is illustrated with the following applications.
\begin{itemize}
  \item
  For independent sets of a graph with $n$ nodes, maximum degree $\Delta$, and parameter $\lambda$, the method allows $O(n)$ expected sampling time when $\lambda < 1.13/(\Delta - 2)$, beating the previously best known linear time algorithm which required $\lambda< 1/(\Delta - 1)$.
  \item
  For the continuous state space autonormal model on a graph with $n$ nodes and parameter $\beta$, the method allows $O(n)$ expected sampling time when $\beta < (\Delta - 1)\exp(-\beta)$, beating the previously best known linear time algorithm which required $\beta < \Delta\exp(-\beta)$.
  \item
  For the random cluster model on a graph with $n$ nodes and maximum degree $\Delta$, and parameters $p$ and $q$, the method allows $O(n)$ expected sampling time when $p < 1/\Delta$.
\end{itemize}

Of course, if the expected running time is allowed to be $O(n \ln (n))$ rather than linear, other
methods with wider parameter ranges exist.
For instance, in~\cite{huber2000a} it was shown using bounding chains and coupling from the past~\cite{proppw1996} how to sample
independent sets in $O(n \ln(n))$ time when $\lambda<2/(\Delta - 2)$.  But for $O(n)$ algorithms,
\prar{} appears to have the widest bounds.

The rest of the paper is organized as follows.  The next section presents the applications, followed by a section describing the theory of the protocol.  Then Section~\ref{SEC:apply} applies the protocol to the different applications, proving the running time results listed above. Section~\ref{SEC:popping} then shows how this
method can be viewed as a generalization of the popping method of Wilson for 
sampling directed rooting spanning trees of a graph.

\section{Applications}

Each of the applications will be described as an unnormalized density $w(x)$ with respect to a base finite measure $\mu$ from which it is easy to sample from.

\subsection{Independent Sets of a Graph}

Consider coloring the nodes of a graph $(V,E)$ with either 0 or 1, so the state space is $\Omega = \{0,1\}^V$.
The underlying measure we use has a parameter $\lambda \geq 0$ that controls the average number of nodes in a randomly drawn coloring:
\begin{equation}
\label{EQN:indsetmu}
\mu(\{x\}) = \lambda^{\sum_i x(i)}.
\end{equation}
Say that $X$ is a Bernoulli random variable with
mean $p$ (write $X \sim \berndist(p)$) if 
$\prob(X = 1) = p$ and $\prob(X = 0) = 1-p$.  To sample from $\mu$, 
generate $X(v) \sim \berndist(\lambda/(1+\lambda))$ independently for 
all $v \in V$.

We use the convention $0^0 = 1$ so that when $\lambda = 0$ the only state with positive measure is the all 0 coloring, while as $\lambda$ grows the measure favors labelings with more 1's.

An {\em independent set} of a graph is a labeling such that no two adjacent nodes receive a 1.  Hence the density of interest is
\[
w(x) = \prod_{\{i,j\} \in E} (1 - x(i)x(j)).
\]
Note that for all $x \in \Omega$, $w(x) \in \{0,1\}$ as well.

\subsection{Strauss process}

The Strauss process~\cite{strauss1975} extends the 
independent set model by allowing two adjacent nodes
to both be labeled 1, but assigns a penalty factor of
$\alpha \in [0,1]$ when this occurs.  The underlying
measure $\mu$ still is defined by~\eqref{EQN:indsetmu},
now the density becomes
\[
w(x) = \prod_{\{i,j\}\in E} \alpha^{x(i)x(j)}.
\]
With the convention that $0^0 = 1$, then 
when $\alpha = 0$ this is just the independent 
set density.

\subsection{Autonormal}

The autonormal model~\cite{besag1974} is a continuous extension of the Ising model~\cite{ising1925} used for modeling images and other spatial experiments.  For a graph $G = (V,E)$, the state space is 
$\Omega = [0,1]^V$, the underlying measure $\mu$ has density
\begin{equation}
\label{EQN:autonormal}
f(x) = \exp\left(-J \sum_{i} (x(i) - y(i))^2 \right)
\end{equation}
with respect to Lebesgue measure.  Here
$J \geq 0$ 
and $y(i)$ are known parameters.  To sample from
$\mu$, each node label $X(i)$ 
is independently  
normal with mean $y(i)$, variance
$(2J)^{-1}$, and conditioned to lie in $[0,1]$.
The density is
\[
w(x) = \exp\left(-\sum_{\{i,j\} \in E} \beta \cdot (x(i) - x(j)^2\right),
\]
where $\beta$ is another constant.

\subsection{Random Cluster}

The random cluster model~\cite{fortuink1972} is another way of viewing the Potts model~\cite{potts52}, which is itself a different extension of the Ising model~\cite{ising1925}.  Given a graph $G = (V,E)$ the state space is now a coloring of edges with 0 and 1.  Given $x \in \Omega = \{0,1\}^E$, and parameter $p$, the underlying measure $\mu$ is
\[
\mu(\{x\}) = p^{\sum_{e \in E} x(e)} (1 - p)^{\sum_{e \in E} (1 - x(e))}.
\]
Again this measure $\mu$ is 
easy to sample from: draw (for all $e \in E$)  
$X(e) \sim \berndist(p)$ independently.

A second parameter $q$ controls the density $w$ through the use of $c(x)$, which is
the number of connected components in the graph using only the edges with $x(e) = 1$.
\begin{equation}
\label{EQN:randomcluster}
w(x) = q^{c(x)}.
\end{equation}
When $q = 2$, a sample from the random cluster model can be transformed into a sample from the Ising model in linear time~\cite{fortuink1972}, and for $q$ an integer greater than 2, it can be transformed in linear time into a sample from the Potts model.

\begin{figure}[ht]
\begin{center}
\begin{tikzpicture}
  \begin{scope}[scale=1]
  \draw[fill=black,thick] (0,0) -- (0,1);
  \draw[dotted] (0,0) grid (1,1);
  \foreach \x in {0,1}
    \foreach \y in {0,1}
      \draw[fill=white,thick] (\x,\y) circle(5pt);
  \end{scope}
\end{tikzpicture}
\end{center}
\caption{An illustration of the random cluster model on a 2 by 2 square lattice.  
Here the state has measure $p(1-p)^3$ since one edge is
included and three edges are not.  The density of this state is $q^3$ since the edges of the state
break the nodes into 3 connected components.  With respect to counting measure, this state has
density $p(1-p)^3 q^3$.}
\end{figure}
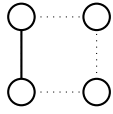

\section{Theory}

The well-known acceptance rejection method requires two properties to hold:
\begin{enumerate}
  \item  
  It is easy to sample from the underlying measure $\mu$.
  \item
  The target density is bounded by a known constant $M$.
\end{enumerate}
When these hold (and they do for each of the applications of the previous section), the acceptance rejection (\ar{})
technique is as follows.

\begin{center}
\setcounter{line}{1}
\begin{tabular}{rl}
\toprule 
\multicolumn{2}{l}{{\tt AR} 
  \quad {\em Output: } $X$} \\
\midrule 
\gk & Randomly draw $X \leftarrow \mu$ \\
\gk & Draw $U$ uniformly from $[0,1]$ \\
\gk & If $U < w(X)/M$ output $X$ and quit \\
\gk & Else \\
\gk & \hspace*{1em} $X \leftarrow {\tt AR}$, output $X$ and quit \\
\bottomrule
\end{tabular}
\end{center}
It is well-known~\cite{fishman1994,huber2015b} that the output is a draw
from density $w(x)$ with respect to measure $\mu$.

Now suppose that $x \in C^D$ for color set $C$ and dimension set $D$ (as it is for all of our applications.)
Let $(D_1,D_2)$ be a partition of the dimensions $D$.  Note that for each of our applications, the product
nature of $w(x)$ means that we can write
\[
w(x) = w_1(x(D_1)) w_2(x(D_2) w_{12}(x(D_1 \cup D_2)),
\]
For instance, for the independent sets model, if $(V_1,V_2)$ is a partition of the nodes $V$, then
\[
w(x) = \left[\prod_{\{i,j\} \in V_1} (1-x(i)x(j))\right]\left[\prod_{\{i,j\} \in V_2} (1-x(i)x(j))\right]
  \left[\prod_{i \in V_1,j \in V_2} (1-x(i)x(j))\right]
\]
and $w(x)$ has been nicely factored into $w_1$, $w_2$, and $w_{12}$.

When for $D_1$ and $D_2$ there is an $M$ such that $w_{12}(x) \leq M$,
fully recursive \ar{} is as follows.
\begin{center}
\setcounter{line}{1}
\begin{tabular}{rl}
\toprule 
\multicolumn{2}{l}{{\tt Fully\_Recusrive\_AR} \quad {\em Input: } $S \subseteq D$ 
  \quad {\em Output: } $X(S)$} \\
\midrule
\gk & Parition $S$ into $D_1$ and $D_2$ \\
\gk & $X(D_1) \leftarrow {\tt Fully\_Recursive\_AR}(D_1)$, $X(D_2) \leftarrow {\tt Fully\_Recursive\_AR}(D_2)$\\
\gk & Draw $U$ uniformly from $[0,1]$ \\
\gk & If $U < w_{12}(X)/M$ output $X$ and quit \\
\gk & Else \\
\gk & \hspace*{1em} $X \leftarrow {\tt Fully\_Recursive\_AR}(S)$, output $X(S)$ and quit \\
\bottomrule
\end{tabular}
\end{center}


\begin{lemma}
  Fully recursive \ar{} generates output $X$ according to $w$ with respect to $\mu$.
\end{lemma}

This is similar to the self-reducibility of Jerrum, Valiant, and Vazirani~\cite{jerrumvv1986}.  More recently,
this approach has also been formalized in~\cite{huber2015b}.  The proof follows
immediately from the correctness of basic \ar{} and an induction on the number of times line 2 has
been used.

\subsection{The \prar{} protocol}

Suppose we desire only to know $X(v)$ for a single node $v$.  Then first
generate $X(v)$ and $X(V \setminus \{v\})$ separately, and then choose to 
accept or reject the combination.  The key idea of \prar{} is that sometimes
this does not require that we generate $X(V \setminus \{v\})$.  For example,
in the independent set model, if $X(v) = 0$, then we accept no matter what
$X(V \setminus \{v\})$ is.  Even if $X(v) = 1$, we do not need to know the entirely
of $X(V \setminus \{v\})$, only those neighbors of $v$.




The general protocol looks like this.
\begin{center}
\setcounter{line}{1}
\begin{tabular}{rl}
\toprule 
\multicolumn{2}{l}{{\tt Partially\_Recursive\_AR}} \\
\multicolumn{2}{l}{{\em Input:  } $S$ and $D$ where $S \subseteq D$} \\
\multicolumn{2}{l}{{\em Output: } $[X(S'),S']$ where $X$ is a draw from $w$ over dimensions $D$ and 
$S \subseteq S'$} \\
\midrule 
\gk & If $S = \emptyset$ then output $[\emptyset,\emptyset]$ \\
\gk & Else \\
\gk & \hspace*{1em} Let $d$ be any element of $S$, let $D_1 \leftarrow \{d\}$ \\
\gk & \hspace*{1em} Draw $X(d)$ randomly from $w_1$ over $\mu$ \\
\gk & \hspace*{1em} Draw $U$ uniformly from $[0,1]$ \\
\gk & \hspace*{1em} If $U \leq \min_{x(D \setminus \{d\})} w_{12}(X(D_1),x(D \setminus D_1))/M$ \\
\gk & \hspace*{2em}   Draw $[X(S''),S''] \leftarrow 
                            {\tt Partially\_Recursive\_AR}(S \setminus \{d\},D \setminus \{d\})$ \\
\gk & \hspace*{2em}   $S' \leftarrow S'' \cup \{d\}$.  
   Output $[X(S'),S']$ and quit \\
\gk  & \hspace*{1em} Else \\
\gk & \hspace*{2em} Let $D_3$ be the smallest set such that $w_{12}(X(D_1 \cup D_3))$ determines $w_{12}(X)$ \\
\gk & \hspace*{2em} $[X(S''),S''] \leftarrow {\tt Partially\_Recursive\_AR}(D_3,D \setminus \{d\})$ \\
\gk & \hspace*{2em} If $U \leq w_{12}(X(D_1 \cup D_3))/M$, then 
  $S' \leftarrow S'' \cup \{d\}$, output $[X(S'),S']$ \\
 \gk & \hspace*{2em} Else let $[X(S'),S] \leftarrow {\tt Partially\_Recursive\_AR}(S,D)$
  and output $[X(S'),S']$ \\
\bottomrule
\end{tabular}
\end{center}

\begin{lemma}
Suppose {\tt Partially\_Recursive\_AR}$(S,D)$ terminates with probability 1 with output
$(X(S),D)$.  
Then $X(S)$ has the same distribution as the coordinates of $S$ from $X$ where 
$X$ is a complete draw from density $w$ with respect to $\mu$ over dimensions $D$.
\end{lemma}

\begin{proof}
  The Fundamental Theorem of Perfect Simulation~\cite{huber2015b,huberPrePrint7} says that if the 
  algorithm terminates with probability 1, then we assume that lines 1, 6, 11, and 13 have the 
  correct output when proving that the overall algorithm is correct.  

  With that assumption, lines 7 and 11 are drawing a state $X$, and then reporting $X(S'')$ for
  some set of dimensions $S''$ that contain $S$.  So the algorithm's output has the same distribution as
  {\tt Fully\_Recursive\_AR}, but not all 
  coordinates of $X$ are reported in the output.  
\end{proof}

It is the analysis of the running time where we see a Dobrushin like condition.

\begin{lemma}
Let $c_1 < 1$ be the probability that line 13 is reached.  Let $p(d)$ denote the probability
that line 10 is reached given $D_1 = d$. 
Suppose for all $d \in D$, there exists $B$ with $\mean[\#(D_3)|d] \leq B$, and that 
$\mean[\#(D_3)|d]p(d) \leq c_2 <  1$.  Given $S$
then the expected number of times line 4 is called (or equivalently line 5) is at most
$(1-c_1)^{-1} (1 - c_2)^{-1} \#(S)$.
\end{lemma}

\begin{proof}
  Consider two temporary changes to the algorithm.  First, at line 13, change it to Else, output
  $[\emptyset,\emptyset]$ and quit.
  Second, at line 3, add $S \leftarrow S \setminus \{d\}$
  to the end. Because of our previous change to line 13, this does not alter the algorithm any further.
  
  Consider the expected time needed to either reach our new line 13, or to quit.  Then there might
  be several recursive calls to {\tt Paritall\_Recursive\_AR} over the run of the algorithm.  Let 
  $N_t$ be the number of dimensions in the union of the $S$ sets after $t$ recursive calls to the 
  algorithm.  Then $N_0 = \#(S)$.  
  
  Now consider $N_{t+1}$ given $N_t$.  Two things can happen.  If lines 7 and 8 activate, then 
  $N_{t+1} = N_t - 1$.  If lines 10 and 11 activate, then $N_{t+1} = N_t  - 1 + \#(D_3)$.  This second event
  happens with probability $p(d)$.  So
  $\mean[N_{t+1}|N_t,d] = N_t - 1 + p(d) \mean[\#(D_3)|d].$  
  
  When this right hand side is at most $N_t - (1 - c_2)$, 
  standard martingale theory (see for instance~\cite{durrett2005})
  gives that the expected amount of time for the $N_t$ process to reach 0 is $N_0/(1 - c_2)$.
  
  Now consider what happens if line 13 is changed back.  Then the expected time between reaching line
  13 is $\#(S)(1 - c_2)^{-1}$.  Each time through the algorithm there is at least a $1 - c_1$ chance
  line 13 is not executed and the algorithm does not recurse further.  Therefore, the total expected
  number of steps taken is at most $(1-c_1)^{-1}(1-c_2)^{-1}\#(S)$.
\end{proof}

\section{Applications}
\label{SEC:apply}

\subsection{Independent sets of a graph}

Independent sets of a graph label each node either 0 or 1.  In this case where there are only two colors, 
the general \prar{} can be
simplified.  If the label of a node is 0, then the node is immediately accepted.  But if the label is 1,
then all the neighbors of the node must also be 0.  If a neighbor is 1, then we can immediately track down
its neighbors to ensure that they are all 0, and so on.  

So if a neighbor of the end of the backbone is labeled 1, that increases the length of the backbone by 1.  
On the other hand, if all the neighbors of the end of the backbone at labeled 0, then the node that is the 
end of the backbone accepts its label of 1.  If the backbone had length 1, then we are done.  Otherwise, that
meant the the second to the end node in the backbone rejects--because it comes into conflict with the end of
the backbone that is labeled 1.  Therefore the last two nodes in backbone are removed from the backbone and
return to being unresolved.  The process then starts over again.

This gives us a connected sequence of nodes labeled 1 that grows and shrinks in length.  Call this sequence
the backbone of the current state, and the method {\tt Backbone\_PRAR}.  See Figure~\ref{FIG:backbone} for 
an illustration.

\begin{center}
\begin{figure}[ht]
\begin{tikzpicture}
  \begin{scope}[scale=0.7]
  \draw[dotted] (0,0) grid (2,2);
  \foreach \x in {0,1,2}
    \foreach \y in {0,1,2}
      \draw[fill=black!20!white,thick] (\x,\y) circle(5pt);
  \draw (1,-.8) node {Step 1};
  \draw[fill=black,thick] (0,2) circle(5pt);
  \end{scope}
  \begin{scope}[scale=0.7,xshift=4cm]
  \draw[dotted] (0,0) grid (2,2);
  \foreach \x in {0,1,2}
    \foreach \y in {0,1,2}
      \draw[fill=black!20!white,thick] (\x,\y) circle(5pt);
  \draw[fill=black,thick] (0,2) circle(5pt);
  \draw[fill=white,thick] (1,2) circle(5pt);
  \draw (1,-.8) node {Step 2};
  \end{scope}
  \begin{scope}[scale=0.7,xshift=8cm]
  \draw[dotted] (0,0) grid (2,2);
  \foreach \x in {0,1,2}
    \foreach \y in {0,1,2}
      \draw[fill=black!20!white,thick] (\x,\y) circle(5pt);
  \draw (1,-.8) node {Step 3};
  \draw[fill=black,thick] (0,2) circle(5pt);
  \draw[fill=white,thick] (1,2) circle(5pt);
  \draw[fill=black,thick] (0,1) circle(5pt);
 \end{scope}
  \begin{scope}[scale=0.7,xshift=12cm]
  \draw[dotted] (0,0) grid (2,2);
  \foreach \x in {0,1,2}
    \foreach \y in {0,1,2}
      \draw[fill=black!20!white,thick] (\x,\y) circle(5pt);
  \draw (1,-.8) node {Step 4};
  \draw[fill=black,thick] (0,2) circle(5pt);
  \draw[fill=white,thick] (1,2) circle(5pt);
  \draw[fill=black,thick] (0,1) circle(5pt);
  \draw[fill=black,thick] (1,1) circle(5pt);
  \end{scope}
  \begin{scope}[scale=0.7,xshift=16cm]
  \draw[dotted] (0,0) grid (2,2);
  \foreach \x in {0,1,2}
    \foreach \y in {0,1,2}
      \draw[fill=black!20!white,thick] (\x,\y) circle(5pt);
  \draw (1,-.8) node {Step 5};
  \draw[fill=black,thick] (0,2) circle(5pt);
  \draw[fill=white,thick] (1,2) circle(5pt);
  \draw[fill=black,thick] (0,1) circle(5pt);
  \draw[fill=black,thick] (1,1) circle(5pt);
  \draw[fill=black,thick] (2,1) circle(5pt);
\end{scope}
  \begin{scope}[scale=0.7,xshift=20cm]
  \draw[dotted] (0,0) grid (2,2);
  \foreach \x in {0,1,2}
    \foreach \y in {0,1,2}
      \draw[fill=black!20!white,thick] (\x,\y) circle(5pt);
  \draw (1,-.8) node {Step 6};
  \draw[fill=black,thick] (0,2) circle(5pt);
  \draw[fill=white,thick] (1,2) circle(5pt);
  \draw[fill=black,thick] (0,1) circle(5pt);
  \draw[fill=black,thick] (1,1) circle(5pt);
  \draw[fill=black,thick] (2,1) circle(5pt);
  \draw[fill=black,thick] (2,0) circle(5pt);
  \end{scope}
  \begin{scope}[scale=0.7,xshift=0cm,yshift=-4cm]
  \draw[dotted] (0,0) grid (2,2);
  \foreach \x in {0,1,2}
    \foreach \y in {0,1,2}
      \draw[fill=black!20!white,thick] (\x,\y) circle(5pt);
  \draw (1,-.8) node {Step 7};
  \draw[fill=black,thick] (0,2) circle(5pt);
  \draw[fill=white,thick] (1,2) circle(5pt);
  \draw[fill=black,thick] (0,1) circle(5pt);
  \draw[fill=black,thick] (1,1) circle(5pt);
  \draw[fill=black,thick] (2,1) circle(5pt);
  \draw[fill=black,thick] (2,0) circle(5pt);
  \draw[fill=white,thick] (1,0) circle(5pt);
\end{scope}
  \begin{scope}[scale=0.7,xshift=4cm,yshift=-4cm]
  \draw[dotted] (0,0) grid (2,2);
  \foreach \x in {0,1,2}
    \foreach \y in {0,1,2}
      \draw[fill=black!20!white,thick] (\x,\y) circle(5pt);
  \draw (1,-.8) node {Step 8};
  \draw[fill=black,thick] (0,2) circle(5pt);
  \draw[fill=white,thick] (1,2) circle(5pt);
  \draw[fill=black,thick] (0,1) circle(5pt);
  \draw[fill=black,thick] (1,1) circle(5pt);
  \end{scope}
  \begin{scope}[scale=0.7,xshift=8cm,yshift=-4cm]
  \draw[dotted] (0,0) grid (2,2);
  \foreach \x in {0,1,2}
    \foreach \y in {0,1,2}
      \draw[fill=black!20!white,thick] (\x,\y) circle(5pt);
  \draw (1,-.8) node {Step 9};
  \draw[fill=black,thick] (0,2) circle(5pt);
  \draw[fill=white,thick] (1,2) circle(5pt);
  \draw[fill=black,thick] (0,1) circle(5pt);
  \draw[fill=black,thick] (1,1) circle(5pt);
  \draw[fill=white,thick] (2,1) circle(5pt);
  \end{scope}
  \begin{scope}[scale=0.7,xshift=12cm,yshift=-4cm]
  \draw[dotted] (0,0) grid (2,2);
  \foreach \x in {0,1,2}
    \foreach \y in {0,1,2}
      \draw[fill=black!20!white,thick] (\x,\y) circle(5pt);
  \draw (1,-.8) node {Step 10};
  \draw[fill=black,thick] (0,2) circle(5pt);
  \draw[fill=white,thick] (1,2) circle(5pt);
  \draw[fill=black,thick] (0,1) circle(5pt);
  \draw[fill=black,thick] (1,1) circle(5pt);
  \draw[fill=white,thick] (2,1) circle(5pt);
  \draw[fill=white,thick] (1,0) circle(5pt);
  \end{scope}
  \begin{scope}[scale=0.7,xshift=16cm,yshift=-4cm]
  \draw[dotted] (0,0) grid (2,2);
  \foreach \x in {0,1,2}
    \foreach \y in {0,1,2}
      \draw[fill=black!20!white,thick] (\x,\y) circle(5pt);
  \draw (1,-.8) node {Step 11};
  \draw[fill=black,thick] (0,2) circle(5pt);
  \draw[fill=white,thick] (1,2) circle(5pt);
  \end{scope}
  \begin{scope}[scale=0.7,xshift=20cm,yshift=-4cm]
  \draw[dotted] (0,0) grid (2,2);
  \foreach \x in {0,1,2}
    \foreach \y in {0,1,2}
      \draw[fill=black!20!white,thick] (\x,\y) circle(5pt);
  \draw (1,-.8) node {Step 12};
  \draw[fill=black,thick] (0,2) circle(5pt);
  \draw[fill=white,thick] (1,2) circle(5pt);
  \draw[fill=white,thick] (0,1) circle(5pt);
  \end{scope} 
\end{tikzpicture}
\caption{Backbone \prar{} for independent sets. Here white stands for label 0 and black for label 1.  
For instance, at step 7 the bottom center node was labeled 0.  This meant that the bottom right node was accepted,
which then meant that the right center node was rejected, so both were removed at Step 8.  By step 12 the backbone 
has resolved and the upper left labeling of 1 is accepted.}
\label{FIG:backbone}
\end{figure}
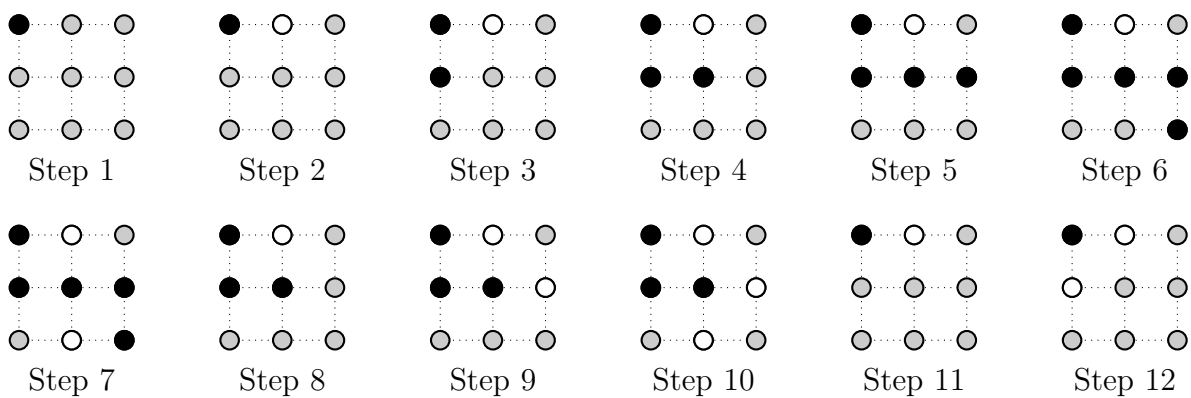
\end{center}

This idea of a backbone works on other applications with two colors, such as the Ising model and the Strauss
model.

\begin{lemma}
The backbone method for independent sets
only requires an independent, identically distributed stream of Bernoulli random
variables with mean $\lambda/(1+\lambda)$.
Let $\Delta$ be the maximum degree of the graph and 
\[
\gamma = \frac{\lambda}{1+\lambda}\sum_{i=1}^{\Delta - 1} \left(1 - \frac{\lambda}{(1+\lambda)^{\Delta}}\right)^{i-1}.
\]
If $\gamma < 1$, then the expected number of Bernoulli's needed to resolve one node is bounded above by  
$C = 1+\lambda \Delta (\Delta - 1)/(1-\gamma)$.  Since this is a constant with respect to the size of the graph,
the time needed to generate a sample over the entire graph is at most $C n$.
\end{lemma}

\begin{proof}
  When the backbone resolves, it resolves either as an acceptance or a rejection.  The number of resolutions
  until acceptance is stochastically bounded above by a geometric random variable with parameter $1/(1+\lambda)$.
  Therefore, the expected number of resolutions of the backbone is at most $(1+\lambda) - 1 = \lambda$.
  
  How many draws does it take to resolve a node?  It certainly takes one to determine the label for the node.
  So if $T$ is the resolution time, then $\mean[T] \leq 1 + \Delta R$, where $R$ is an upper bound on the 
  expected resolution time for subsequent nodes.  These subsequent nodes have maximum degree at most $\Delta - 1$.
  The first neighbor takes at most $R$ expected time.
  
  On to the second neighbor.  Note that the second neighbor only needs to be resolved if the first neighbor resolved
  to a 0.  If the first neighbor had resolved to a 1, then the original node is rejected, there is no need for 
  further action.  The chance that a neighbor resolves to 1 is at least 
  $c = [\lambda/(1+\lambda)][1/(1+\lambda)]^{\Delta - 1}$.  So the second neighbor is only activated with
  probability $1 - c$.  
  
  Similarly, the third neighbor can also take time 
  distributed as $R$, but only if the first two neighbors fail.  Adding this up over the (up to $\Delta - 1$) neighbors gives
  \[
  R \leq 1 + \frac{\lambda}{1+\lambda} \sum_{i=1}^{\Delta - 1} R \left(1 - \left[\frac{\lambda}{1+\lambda}\right]
   \left[\frac{1}{1+\lambda}\right]^{\Delta - 1}\right)^{i-1},
  \]
  which gives the result.
\end{proof}

It is straightforward to verify that for 
$\lambda < 1.13/(\Delta - 2)$ we have $\gamma < 1$. 
This gives 
the result presented in Section~\ref{SEc:introduction}.

Call the value of $\lambda$ where $\gamma = 1$ the {\em critical value} of $\lambda$, and denote it
$\lambda_c$.  For $\lambda < \lambda_c$, the algorithm is guaranteed to generate samples in polynomial time.
For $\lambda \geq \lambda_c$, the algorithm might operate in polynomial time, or it might not.

\begin{figure}[ht]
\begin{center}
\begin{tabular}{rlll}
$\Delta$ & $\lambda_c$ for RR & $\lambda_c$ for CoA
 & $\lambda_c$ for PRAR \\
\midrule
3 & 0.2 &            0.5          & 1.13224\ldots \\
4 & 0.142857\ldots & 0.333\ldots  & 0.57833\ldots \\
5 & 0.111111\ldots & 0.25         & 0.29315\ldots \\
\end{tabular}
\caption{Range of critical $\lambda$ compared to 
$\Delta$.  Since expected sampling time is guaranteed linear for $\lambda < \lambda_c$, higher
is better.  Column for $\lambda_c$ for the Randomness Recycler (RR) comes from Theorem 1 of~\cite{huber2000b}.  Column for 
$\lambda_c$ for Clan of Ancestors (CoA) comes from a branching process analysis.}
\end{center}
\end{figure}

\subsection{Strauss process}

Now consider what happens in the Strauss process.
Again consider drawing an initial node and accepting
or rejecting based upon the neighbors of the node.

If the node is labeled 0, then we always accept as 
before.  If the node is labeled 1, then the chance of accepting is $\alpha$ raised to the power of each of the neighbors of the node labeled 1.  

Another way to view this, is to, for each edge adjacent to the original node, draw a $\berndist(\alpha)$ random
variable.  If this variable is 1, then the neighbor does not matter.  If this variable is 0, then the neighbor does matter, and recursion needs to be used to find out its value.

This changes the expected number of neighbors to be
considered from $\Delta - 1$ to $(\Delta - 1)\alpha$.  The chance of a node being labeled 1 is strictly greater in the recursion, therefore, the new value of $\gamma_\alpha$ is at most
\[
\gamma_\alpha = \frac{\lambda}{1+\lambda} \alpha(\Delta - 1).
\]
Note that $\gamma_\alpha < 1$ is equivalent to 
$\lambda < 1/[\alpha(\Delta - 1) -1]$, so 
a similar argument to the previous section gives the
following result.

\begin{lemma}
When $\lambda < [\alpha(\Delta - 1) -1]^{-1},$ then \prar{} generates
a sample from the Strauss process using $O(n)$ 
expected number of Bernoulli draws.
\end{lemma}

\subsection{Autonormal}

As with the independent set density, the idea is to consider what happens when a single node
is drawn, the rest of the graph is drawn, and the combination is either accepted or rejected.
The value that $U$ from line 6 in the general protocol must fall below is the product of the 
chance of rejecting because of each of the neighboring edges.

That is, in order to determine if 
node $v$ can be combined with the state of $x(V \setminus \{v\})$, it is necessary to 
independently draw a uniform $[0,1]$ random variable for each of the neighbors $w$ of $v$.  
Only if every edge accepts can the entire state be said to accept.

If the single node $v$ is assigned value $x(v)$, the chance of rejecting based on $w$ is at most 
\[
\exp(-\beta(x(v) - x(w))^2) \leq \exp(-\beta\max\{1-x(i),x(i)\}^2),
\]
where the upper bound on the right hand side does not depend on $x(w)$!

Therefore, only if $U_w \geq \exp(-\beta\max\{1-x(i),x(i)\}^2)$ does the value $x(w)$ need to 
be determined.  Therefore, the expected number of neighbors of the original node $v$ that need to 
be found out is bounded above by 
\[
\Delta(1-\exp(-\beta\max\{1-x(i),x(i)\}^2)).
\]
The original node $v$ might have had up to $\Delta$ neighbors, but subsequent nodes will only have at most 
$\Delta - 1$ neighbors.  Therefore, if $(\Delta - 1)(1-\exp(-\beta\max\{1-x(i),x(i)\}^2)) < 1$, then 
on average the number of new nodes to consider generated by a node will be negative, and \prar{} will terminate
after a finite number of steps.  This argument yields the following lemma.

\begin{lemma}
  Let $\gamma = (\Delta - 1)\exp(-\beta)$.  If $\gamma < 1$, then \prar{} can generate a sample from the Autonormal
  model in $O(n)$ random choices on a graph with $n$ nodes and maximum degree $\Delta$.
\end{lemma}

A similar analysis for the Clan of Ancestors approach requires $\Delta\exp(-\beta) < 1$, therefore the parameter
range with guaranteed performance here is slightly wider.

\subsection{Random Cluster}

Suppose that $q > 1$ in the random cluster model with density given by~\eqref{EQN:randomcluster}.
Then the density of a configuration gains a factor of $q$ for each connected component in the 
graph.

Viewed as a penalty, this means that for a particular edge, there is a penalty of $1/q$ if 
it connects two previously disconnected components in the rest of the graph.  In other words,
if $c = (1/q)^n$, where $n$ is the number of nodes in the graph, then the density can be written 
as 
\[
w_2(x) = c w(x) = (1/q)^{n - c(x))}.
\]

Now consider a single edge $\{i,j\}$ chosen from the underlying measure.  If the edge is labeled 0
(which happens with probability $1-p$), then the edge is accepted regardless of the rest
of the components.  If the edge is labeled 1 (which happens with probability $p$), then 
the edge is accepted with probability $1/q$ if the rest of the state does not connect
nodes $i$ and $j$, and accepted with probability 1 otherwise.

Suppose that $x(\{i,j\})$ is randomly chosen to be 1.
Then draw $U$ uniformly over $[0,1]$.  If $U < 1/q$, then the edge $\{i,j\}$
is accepted regardless of the rest of the state and the edge value becomes known.  Otherwise,
recursion must be used to determine enough of the rest of the state in order to determine if
the nodes $\{i,j\}$ are connected or not.

The first edge $\{i,j\}$ can be connected to as many as $2(\Delta - 1)$ different edges, each of which must be considered to determine if $x(\{i,j\}) = 1$.  Let $e$ be one of these edges (either
of the form $\{i,k\}$ or $\{j,k\}$ for some nodes $k$.)

If $e$ is chosen to be labeled 0 in the recursion, then it is removed from the ``needs a value'' list.  Otherwise, it is adjacent to at most $\Delta - 1$ new edges that might
need a value.  Hence the expected change of the size of the number of edges that need
values is bounded above by
\[
(-1)(1-p) + p(\Delta - 1) = \Delta p - 1.
\]
If this quantity is less than 0, then the number of edges to be considered drops (on average)
at each step, and the expected number of steps is bounded.  This gives the following lemma.

\begin{lemma}
  Let $G$ be a graph with $n$ nodes, $m$ edges, and maximum degree $\Delta$. 
  If $p < 1/\Delta$, then the expected number of random choices made in the algorithm is
  $O(m)$.
\end{lemma}

\section{Connection to popping}
\label{SEC:popping}

This \prar{} protocol can be viewed as an extension of 
the cycle popping algorithm of Wilson for uniformly 
generating from rooted trees in a directed graph.

Consider a graph $(V,E)$, and construct
a directed graph by taking each edge $\{i,j\}$ and
adding directed arcs $(i,j)$ and $(j,i)$.
Designate a special node $r \in V$
known as the root of the graph.  
Then let a configuration consist
of a $\{0,1\}$ labeling of arcs such that 
every node $v \neq r$ has exactly one outgoing edge 
$(v,w)$ labeled 1.  

Let $\mu$ be the underlying
measure that is uniform over all such labelings.
Then it is easy to sample from $\mu$:  independently
for each node $i$ select a neighbor $j$ uniformly at
random and label that outgoing arc 1, and all other
outgoing arcs from $i$ label 0.

The density $w(x)$ then is 1 if 
every node has a directed path
using edges labeled 1 to the root, and is 0 otherwise.
When $w(x)$ is 1, say that $x$ encodes a 
{\em directed rooted tree} in the original graph.

Then Wilson~\cite{wilson1996} presented a simple
algorithm for generating from $w(x)$ over $\mu$.
Start with an arbitrary node $i$, and uniformly
choose a neighbor $j$.  From $j$ choose neighbor $k$,
and continue until one of two things happens.

If the directed path reaches a node already examined,
then there is a loop.  Erase the loop and continue
onwards.  For instance, if the choices where
$a \rightarrow b \rightarrow c \rightarrow d \rightarrow b$, then $(b,c,d,b)$ forms a loop, and the next step
would start with $a \rightarrow b$.
The other thing that can happen is that the directed
path reaches $r$.  In this case, fix the labeling for 
all nodes along the path.  

Now choose another node in the graph and repeat, choosing
random neighbors, erasing loops as they form, and stopping when reaching either $r$ or a previously fixed node.

This algorithm is just \prar{} with the backbone method applied to $\mu$ and $w$!  At each step of the algorithm
we are recursively moving deeper into the graph in order to determine if acceptance should occur.  If a loop forms, then rejection occurs, and all the edges in the loop are eliminated.  But if the path should encounter a previously fixed path, then acceptance occurs not only for the original edge but all along the path as well.

Wilson gave a proof tailored for his algorithm.  Because this is also a \prar{} algorithm, we immediately have correctness for the loop-erased random walk method of uniformly generating rooted trees.

%

\section*{Acknowledgements}

An earlier version of this technique appears in~\cite{huber2015b} (without the backbone method, connection to popping, and the general running time
bound given here.)  The recent development of this
method was supported by NSF grant DMS-1418495.

\bibliographystyle{plain}

\end{document}